\documentclass{amsart}
\usepackage[foot]{amsaddr}
\usepackage{amssymb}
\usepackage{bigints}
\usepackage[noadjust]{cite}
\usepackage{enumitem}\setlist[enumerate]{itemsep=1ex}
\usepackage[colorlinks,citecolor=blue,urlcolor=blue]{hyperref}
\usepackage{orcidlink}
\usepackage{tikz}\usetikzlibrary{decorations.markings,shapes.misc}

\theoremstyle{definition}
\newtheorem{definition}{Definition}
\numberwithin{definition}{section}
\newtheorem{example}[definition]{Example}
\newtheorem{remark}[definition]{Remark}

\theoremstyle{plain}
\newtheorem{corollary}[definition]{Corollary}
\newtheorem{lemma}[definition]{Lemma}
\newtheorem{proposition}[definition]{Proposition}
\newtheorem{theorem}[definition]{Theorem}
\newtheorem*{maintheorem}{Main Theorem}

\renewcommand{\Diamond}{\rotatebox[origin=c]{45}{$\Box$}}
\AtEndEnvironment{definition}{\hfill\Diamond}
\AtEndEnvironment{example}{\hfill\Diamond}
\AtEndEnvironment{remark}{\hfill\Diamond}

\DeclareMathOperator{\adj}{adj}
\DeclareMathOperator{\Adm}{Adm}
\DeclareMathOperator{\Com}{Com}
\let\d\relax\DeclareMathOperator{\d}{d\!}
\newcommand{\ie}{\emph{i.e.}\ }
\DeclareMathOperator{\init}{in}
\newcommand{\p}{\mathfrak p}
\newcommand{\R}{\mathbb R}
\DeclareMathOperator{\Tub}{Tub}
\DeclareMathOperator{\vol}{vol}

\title{Representations of the Flat Space Wavefunction}
\author[T. Dunaisky]{Tyler Dunaisky\orcidlink{0009-0000-2123-1789}}
\address{Purdue University Department of Mathematics, West Lafayette, IN 47907, USA.}
\email{tdunaisk@purdue.edu}

\begin{document}

\begin{abstract}
From any graph $G$ arises a flat space wavefunction, obtained by integrating a product of propagators associated to the vertices and edges of $G$. This function is a key ingredient in the computation of cosmological correlators, and several representations for it have been proposed. We formulate three such representations and prove their correctness. In particular, we show that the flat space wavefunction can be read off from the canonical form of the cosmological polytope, and we settle a conjecture of Fevola, Pimentel, Sattelberger, and Westerdijk regarding a partial fraction decomposition for the flat space wavefunction. The terms of the decomposition correspond to certain collections of connected subgraphs associated to $G$ and its spanning subgraphs, reflecting the fact that the flat space wavefunction contains information about how $G$ is connected.
\end{abstract}

\maketitle
\tableofcontents

\section{Introduction}

Arkani-Hamed, Benincasa, and Postnikov \cite{AH-B-P} introduced the \emph{wavefunction of the universe}, an object which provides information about the state of the early universe. The wavefunction is built up by contributions from \emph{cosmological correlators} $\Psi_G$, each associated to a graph $G$. Always, $G$ is a finite graph (possibly with multiple edges or edges from a vertex to itself): it can be thought of as a diagram in spacetime describing energy exchanges between particles. 

Computation of the cosmological correlator requires evaluation of the integral
$$\Psi_G(\mathbf{X},\mathbf{Y})=\int_{\R_{\geq0}^n}\psi_G(\mathbf{X}+\boldsymbol{\alpha},\mathbf{Y})\boldsymbol{\alpha}^\varepsilon\d\boldsymbol{\alpha},$$
where $\varepsilon$ is a constant encoding the underlying cosmological model and $\psi_G$ is the \emph{flat space wavefunction} (so called because it agrees with $\Psi_G$ in the ``flat space" limit ${\varepsilon\to-1}$). These integrals prove to be extremely challenging: see \cite{AH-B-H-J-L-P,F-MH,F-P-S-W,G-H,P-S} for some approaches to this problem. Moreover, computation of the function appearing in the integrand is a difficult problem in its own right. In \cite{AH-B-P}, Arkani-Hamed, Benincasa, and Postnikov worked out several examples and observed that $\psi_G$ has unexpectedly nice expressions in terms of linear forms $\ell_T$ associated to the connected subgraphs (``tubes") of $G$. Motivated by this and the underlying physics, they proposed several representations for $\psi_G$, including a ``bulk" representation and a ``boundary" representation. Most notably, they conjectured \cite[Equation 3.5]{AH-B-P} that the \emph{canonical form} of the \emph{cosmological polytope} is of the form
$$\Omega_G=\psi_G\d\mathbf{X}\d\mathbf{Y}.$$
Canonical forms, which are more generally associated to any \emph{positive geometry}, appear in many other contexts and have been studied extensively in recent years \cite{AH-B-L,B-D,Lam,P-T}.

Throughout the cosmological correlator literature, the representations mentioned above have been used implicitly. In this paper, we give proofs that $\psi_G$ indeed has the above canonical form representation, as well as two other representations in terms of certain collections of connected subgraphs (``tubings"), which we interpret as realizations of the bulk and boundary representations. In particular, our bulk representation is an amended version of \cite[Conjecture 3.5]{F-P-S-W} and equivalent to the formula proposed in \cite[Equation 4.1]{Glew}. Along the way, we describe a natural triangulation of the \emph{dual cosmological polytope} (for triangulations of the cosmological polytope, see \cite{J-S-V}) using the language of tubings.

\begin{maintheorem}[Theorems \ref{theorem_bulk}, \ref{theorem_boundary}, \ref{theorem_canonicalform}]
    For any graph $G$, we have the following representations of the flat space wavefunction $\psi_G$:

    \begin{enumerate}[label=(\alph*)]

        \item
        the bulk representation:
        $$\psi_G=\frac{1}{\prod\limits_{e\in E}(2Y_e)}\sum_H\sum_{\mathbf{T}}\frac{(-1)^{|E\setminus E_H|}}{\prod\limits_{T\in\mathbf{T}}\ell_T},$$
        where the sums are over all spanning subgraphs $H=(V,E_H)$ of $G=(V,E)$ and all admissible tubings $\mathbf{T}$ of $H$.

        \item
        the boundary representation:
        $$\psi_G=\sum_{\mathbf{T}}\frac{1}{\prod\limits_{T\in\mathbf{T}}\ell_T},$$
        where the sum is over all complete tubings $\mathbf{T}$ of $G$.

        \item
        the canonical form representation:
        $$\psi_G=\frac{\adj_G}{\displaystyle\prod\limits_{T}\ell_T},$$
        where the product is over all tubes of $G$ and $\adj_G$ is the adjoint polynomial of the dual cosmological polytope.

    \end{enumerate}
\end{maintheorem}

It follows from (c) and Proposition \ref{prop_canonicalform} that the canonical form of the cosmological polytope is indeed $\Omega_G=\psi_G\d\mathbf{X}\d\mathbf{Y}$. From this perspective, we can also interpret (a) and (b) as providing representations of $\Omega_G$.

The paper is laid out as follows. In Section \ref{sec_tubings}, we define admissible and complete tubings and prove some useful lemmas regarding them. In Sections \ref{sec_bulk}, \ref{sec_boundary}, and \ref{sec_canonicalform}, we prove the bulk, boundary, and canonical form representations, respectively. Always, we work with a finite graph $G$ with vertices $V=\{v_1,\dots,v_n\}$ and edges $E$.

\section{Tubings}\label{sec_tubings}

In this section, we define admissible and complete tubings and prove some preliminary results. The main result is Lemma \ref{lemma_bijection}, which allows us to label the tubes in an admissible tubing by the vertices of $G$. The notion of admissible tubings was introduced in \cite{C-D} (though they would call them $n$-tubings), and the notion of complete tubings was introduced in \cite{AH-B-H-J-L-P}.

\begin{definition}\label{def_tubes}
    Fix a graph $G=(V,E)$. A \emph{tube} of $G$ is a connected subgraph $T\subseteq G$. We denote the set of all tubes of $G$ by $\Tub(G)$. We say $T$ is \emph{induced} if it is an induced subgraph of $G$, \ie if $T$ contains all the edges between its vertices. We say tubes $S,T\in\Tub(G)$ are:
    \begin{enumerate}[label=-]

        \item\emph{nested} if $S\subseteq T$ or $T\subseteq S$,

        \item\emph{overlapping} if $S\cap T$ is nonempty, and $S$, $T$ are not nested,

        \item\emph{incompatible} if there is an edge between $S$ and $T$, and $S$, $T$ are not nested.

    \end{enumerate}
    Note that overlapping tubes are automatically incompatible, since we require tubes to be connected.
\end{definition}

\begin{figure}
    \renewcommand{\arraystretch}{1.5}
    \begin{tabular}{ccccccc}
        \begin{tikzpicture}
            \draw (0,0) circle (0.375);
            \draw[fill=black] (-0.375,0) circle (.5mm);
            \draw[fill=black] (0.375,0) circle (.5mm);
            \draw[color=gray] (0:0.175) arc (0:180:0.175);
            \draw[color=gray] (0:0.575) arc (0:180:0.575);
            \draw[color=gray] (180:0.575) arc (-180:0:0.2);
            \draw[color=gray] (0:0.575) arc (0:-180:0.2);
        \end{tikzpicture}
        & &
        \begin{tikzpicture}
            \draw (0,0) circle (0.375);
            \draw[fill=black] (-0.375,0) circle (.5mm);
            \draw[fill=black] (0.375,0) circle (.5mm);
            \draw[color=gray] (-0.375,0) circle (.15cm and .15cm);
            \draw[color=gray] (0:0.175) arc (0:180:0.175);
            \draw[color=gray] (0:0.575) arc (0:180:0.575);
            \draw[color=gray] (180:0.575) arc (-180:0:0.2);
            \draw[color=gray] (0:0.575) arc (0:-180:0.2);
        \end{tikzpicture}
        & &
        \begin{tikzpicture}
            \draw (0,0) circle (0.375);
            \draw[fill=black] (-0.375,0) circle (.5mm);
            \draw[fill=black] (0.375,0) circle (.5mm);
            \draw[color=gray] (0:0.175) arc (0:180:0.175);
            \draw[color=gray] (0:0.575) arc (0:180:0.575);
            \draw[color=gray] (180:0.575) arc (-180:0:0.2);
            \draw[color=gray] (0:0.575) arc (0:-180:0.2);
            \draw[color=gray] (0:0.25) arc (0:-180:0.25);
            \draw[color=gray] (0:0.5) arc (0:-180:0.5);
            \draw[color=gray] (180:0.5) arc (180:0:0.125);
            \draw[color=gray] (0:0.5) arc (0:180:0.125);
        \end{tikzpicture}
        & &
        \begin{tikzpicture}
            \draw (0,0) circle (0.375);
            \draw[fill=black] (-0.375,0) circle (.5mm);
            \draw[fill=black] (0.375,0) circle (.5mm);
            \draw[color=gray] (-0.375,0) circle (.15cm and .15cm);
            \draw[color=gray] (0.375,0) circle (.15cm and .15cm);
        \end{tikzpicture}
        \\
        (a) & & (b) & & (c) & & (d)
    \end{tabular}
\caption{Examples of a) a non-induced tube, b) nested tubes, c) overlapping tubes, d) incompatible (but not overlapping) tubes.}
\label{fig_tubeexamples}
\end{figure}
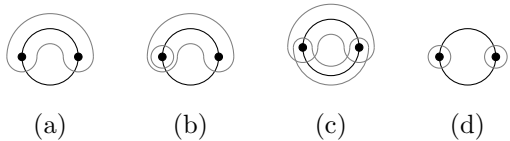

See Figure \ref{fig_tubeexamples} for some examples. We will say that tubes which are not incompatible are \textit{compatible}.

\begin{definition}
    Fix a graph $G$. A collection of tubes $\mathbf T\subseteq\Tub(G)$ is a:
    \begin{enumerate}[label=-]

        \item \emph{complete tubing} if it is a maximal collection of non-overlapping tubes.

        \item \emph{admissible tubing} if it is a maximal collection of compatible induced tubes.

    \end{enumerate}
    We denote the collections of all such tubings by
    \begin{align*}
        \Com(G)&=\{\text{complete tubings of $G$}\},\\
        \Adm(G)&=\{\text{admissible tubings of $G$}\},
    \end{align*}
    for convenience.
\end{definition}

Figure \ref{fig_tubingexamples} lists the complete and admissible tubings for some small graphs. It is easy to see that complete tubings always contain every singleton tube $\{v\}$. Disregarding these, observe there is a correspondence between the complete tubings of the $1$-loop bubble and the admissible tubings of the $2$-site chain, and likewise between the complete tubings of the $3$-site chain and the admissible tubings of the $2$-site chain. Such a correspondence holds in general, as we explain below.

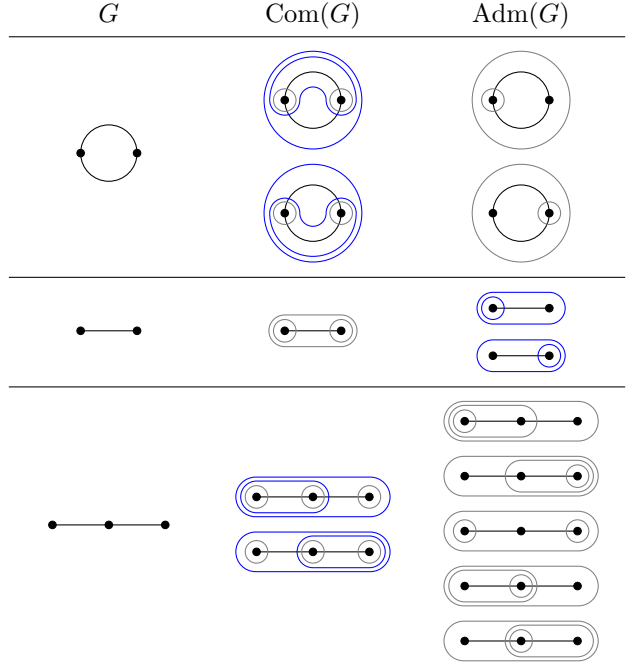
\begin{figure}
    \renewcommand{\arraystretch}{1.5}
    \begin{tabular}{ccc}
        \hspace{0.4in}$G$\hspace{0.4in} & \hspace{0.2in}$\Com(G)$\hspace{0.2in} & \hspace{0.2in}$\Adm(G)$\hspace{0.2in}
        \\ \hline
        \begin{tikzpicture}[baseline=-4]
            \draw (0,0) circle (0.375);
            \draw[fill=black] (-0.375,0) circle (.5mm);
            \draw[fill=black] (0.375,0) circle (.5mm);
        \end{tikzpicture}
        &
        \begin{tabular}{c}\\[-3ex]
            \begin{tikzpicture}
                \draw (0,0) circle (0.375);
                \draw[fill=black] (-0.375,0) circle (.5mm);
                \draw[fill=black] (0.375,0) circle (.5mm);
                \draw[color=gray] (-0.375,0) circle (.15cm and .15cm);
                \draw[color=gray] (0.375,0) circle (.15cm and .15cm);
                \draw[color=blue] (0,0) circle (0.65);
                \draw[color=blue] (0:0.175) arc (0:180:0.175);
                \draw[color=blue] (0:0.575) arc (0:180:0.575);
                \draw[color=blue] (180:0.575) arc (-180:0:0.2);
                \draw[color=blue] (0:0.575) arc (0:-180:0.2);
            \end{tikzpicture}
            \\
            \begin{tikzpicture}
                \draw (0,0) circle (0.375);
                \draw[fill=black] (-0.375,0) circle (.5mm);
                \draw[fill=black] (0.375,0) circle (.5mm);
                \draw[color=gray] (-0.375,0) circle (.15cm and .15cm);
                \draw[color=gray] (0.375,0) circle (.15cm and .15cm);
                \draw[color=blue] (0,0) circle (0.65);
                \draw[color=blue] (0:0.175) arc (0:-180:0.175);
                \draw[color=blue] (0:0.575) arc (0:-180:0.575);
                \draw[color=blue] (180:0.575) arc (180:0:0.2);
                \draw[color=blue] (0:0.575) arc (0:180:0.2);
            \end{tikzpicture}
        \end{tabular}
        &
        \begin{tabular}{c}\\[-3ex]
            \begin{tikzpicture}
                \draw (0,0) circle (0.375);
                \draw[fill=black] (-0.375,0) circle (.5mm);
                \draw[fill=black] (0.375,0) circle (.5mm);
                \draw[color=gray] (-0.375,0) circle (.15cm and .15cm);
                \draw[color=gray] (0,0) circle (0.65);
            \end{tikzpicture}
            \\
            \begin{tikzpicture}
                \draw (0,0) circle (0.375);
                \draw[fill=black] (-0.375,0) circle (.5mm);
                \draw[fill=black] (0.375,0) circle (.5mm);
                \draw[color=gray] (0.375,0) circle (.15cm and .15cm);
                \draw[color=gray] (0,0) circle (0.65);
            \end{tikzpicture}
        \end{tabular}
        \\ \hline
        \begin{tikzpicture}[baseline=-3]
            \draw[fill=black] (-0.375,0) -- (0.375,0);
            \draw[fill=black] (-0.375,0) circle (.5mm);
            \draw[fill=black] (0.375,0) circle (.5mm);
        \end{tikzpicture}
        &
        \begin{tikzpicture}[baseline=-3]
            \draw[fill=black] (-0.375,0) -- (0.375,0);
            \draw[fill=black] (-0.375,0) circle (.5mm);
            \draw[fill=black] (0.375,0) circle (.5mm);
            \draw[color=gray] (-0.375,0) circle (.15cm and .15cm);
            \draw[color=gray] (0.375,0) circle (.15cm and .15cm);
            \node [draw, color=gray, rounded rectangle, minimum height = 1.2em, minimum width = 4em, rounded rectangle arc length = 180] at (0,0) {};
        \end{tikzpicture}
        &
        \begin{tabular}{c}\\[-3ex]
            \begin{tikzpicture}
                \draw[fill=black] (-0.375,0) -- (0.375,0);
                \draw[fill=black] (-0.375,0) circle (.5mm);
                \draw[fill=black] (0.375,0) circle (.5mm);
                \draw[color=blue] (-0.375,0) circle (.15cm and .15cm);
                \node [draw, color=blue, rounded rectangle, minimum height = 1.2em, minimum width = 4em, rounded rectangle arc length = 180] at (0,0) {};
            \end{tikzpicture}
            \\
            \begin{tikzpicture}
                \draw[fill=black] (-0.375,0) -- (0.375,0);
                \draw[fill=black] (-0.375,0) circle (.5mm);
                \draw[fill=black] (0.375,0) circle (.5mm);
                \draw[color=blue] (0.375,0) circle (.15cm and .15cm);
                \node [draw, color=blue, rounded rectangle, minimum height = 1.2em, minimum width = 4em, rounded rectangle arc length = 180] at (0,0) {};
            \end{tikzpicture}
        \end{tabular}
        \\ \hline
        \begin{tikzpicture}[baseline=-5]
            \draw[fill=black] (-0.75,0) -- (0.75,0);
            \draw[fill=black] (0,0) circle (.5mm);
            \draw[fill=black] (-0.75,0) circle (.5mm);
            \draw[fill=black] (0.75,0) circle (.5mm);
        \end{tikzpicture}
        &
        \begin{tabular}{c}
            \begin{tikzpicture}
                \draw[fill=black] (-0.75,0) -- (0.75,0);
                \draw[fill=black] (0,0) circle (.5mm);
                \draw[fill=black] (-0.75,0) circle (.5mm);
                \draw[fill=black] (0.75,0) circle (.5mm);
                \draw[color=gray] (-0.75,0) circle (.15cm and .15cm);
                \draw[color=gray] (0,0) circle (.15cm and .15cm);
                \draw[color=gray] (0.75,0) circle (.15cm and .15cm);
                \node [draw, color=blue, rounded rectangle, minimum height = 1.2em, minimum width = 4em, rounded rectangle arc length = 180] at (-0.375,0) {};
                \node [draw, color=blue, rounded rectangle, minimum height = 1.5em, minimum width = 6.5em, rounded rectangle arc length = 180] at (0,0) {};
            \end{tikzpicture}
            \\
            \begin{tikzpicture}
                \draw[fill=black] (-0.75,0) -- (0.75,0);
                \draw[fill=black] (0,0) circle (.5mm);
                \draw[fill=black] (-0.75,0) circle (.5mm);
                \draw[fill=black] (0.75,0) circle (.5mm);
                \draw[color=gray] (-0.75,0) circle (.15cm and .15cm);
                \draw[color=gray] (0,0) circle (.15cm and .15cm);
                \draw[color=gray] (0.75,0) circle (.15cm and .15cm);
                \node [draw, color=blue, rounded rectangle, minimum height = 1.2em, minimum width = 4em, rounded rectangle arc length = 180] at (0.375,0) {};
                \node [draw, color=blue, rounded rectangle, minimum height = 1.5em, minimum width = 6.5em, rounded rectangle arc length = 180] at (0,0) {};
            \end{tikzpicture}
        \end{tabular}
        &
        \begin{tabular}{c}\\[-3ex]
            \begin{tikzpicture}
                \draw[fill=black] (-0.75,0) -- (0.75,0);
                \draw[fill=black] (0,0) circle (.5mm);
                \draw[fill=black] (-0.75,0) circle (.5mm);
                \draw[fill=black] (0.75,0) circle (.5mm);
                \draw[color=gray] (-0.75,0) circle (.15cm and .15cm);
                \node [draw, color=gray, rounded rectangle, minimum height = 1.2em, minimum width = 4em, rounded rectangle arc length = 180] at (-0.375,0) {};
                \node [draw, color=gray, rounded rectangle, minimum height = 1.5em, minimum width = 6.5em, rounded rectangle arc length = 180] at (0,0) {};
            \end{tikzpicture}
            \\
            \begin{tikzpicture}
                \draw[fill=black] (-0.75,0) -- (0.75,0);
                \draw[fill=black] (0,0) circle (.5mm);
                \draw[fill=black] (-0.75,0) circle (.5mm);
                \draw[fill=black] (0.75,0) circle (.5mm);
                \draw[color=gray] (0.75,0) circle (.15cm and .15cm);
                \node [draw, color=gray, rounded rectangle, minimum height = 1.2em, minimum width = 4em, rounded rectangle arc length = 180] at (0.375,0) {};
                \node [draw, color=gray, rounded rectangle, minimum height = 1.5em, minimum width = 6.5em, rounded rectangle arc length = 180] at (0,0) {};
            \end{tikzpicture}
            \\
            \begin{tikzpicture}
                \draw[fill=black] (-0.75,0) -- (0.75,0);
                \draw[fill=black] (0,0) circle (.5mm);
                \draw[fill=black] (-0.75,0) circle (.5mm);
                \draw[fill=black] (0.75,0) circle (.5mm);
                \draw[color=gray] (-0.75,0) circle (.15cm and .15cm);
                \draw[color=gray] (0.75,0) circle (.15cm and .15cm);
                \node [draw, color=gray, rounded rectangle, minimum height = 1.5em, minimum width = 6.5em, rounded rectangle arc length = 180] at (0,0) {};
            \end{tikzpicture}
            \\
            \begin{tikzpicture}
                \draw[fill=black] (-0.75,0) -- (0.75,0);
                \draw[fill=black] (0,0) circle (.5mm);
                \draw[fill=black] (-0.75,0) circle (.5mm);
                \draw[fill=black] (0.75,0) circle (.5mm);
                \draw[color=gray] (0,0) circle (.15cm and .15cm);
                \node [draw, color=gray, rounded rectangle, minimum height = 1.2em, minimum width = 4em, rounded rectangle arc length = 180] at (-0.375,0) {};
                \node [draw, color=gray, rounded rectangle, minimum height = 1.5em, minimum width = 6.5em, rounded rectangle arc length = 180] at (0,0) {};
            \end{tikzpicture}
            \\
            \begin{tikzpicture}
                \draw[fill=black] (-0.75,0) -- (0.75,0);
                \draw[fill=black] (0,0) circle (.5mm);
                \draw[fill=black] (-0.75,0) circle (.5mm);
                \draw[fill=black] (0.75,0) circle (.5mm);
                \draw[color=gray] (0,0) circle (.15cm and .15cm);
                \node [draw, color=gray, rounded rectangle, minimum height = 1.2em, minimum width = 4em, rounded rectangle arc length = 180] at (0.375,0) {};
                \node [draw, color=gray, rounded rectangle, minimum height = 1.5em, minimum width = 6.5em, rounded rectangle arc length = 180] at (0,0) {};
            \end{tikzpicture}
        \end{tabular}
    \end{tabular}
\caption{Every complete tubing/admissible tubing of the ``$1$-loop bubble", ``$2$-site chain", and ``$3$-site chain".}
\label{fig_tubingexamples}
\end{figure}

\begin{definition}
    To a graph $G=(V,E)$ with at least one edge, we associate the \emph{line graph} $L(G)$, which has a vertex for each edge $e\in E$ and an edge between any $e,e'\in E$ which share at least one vertex.
\end{definition}

\begin{lemma}\label{lemma_linecorrespondence}
    There is a correspondence
    \begin{center}
        \vspace{-0.15in}
        $$\begin{array}{crc}
            \begin{array}{c} \left\{\begin{array}{c}\text{complete tubings}\\\text{of }G\end{array}\right\} \end{array} & \longleftrightarrow & \begin{array}{c} \left\{\begin{array}{c}\text{admissible tubings}\\\text{of }L(G)\end{array}\right\} \end{array} \\[3ex]
            \mathbf T & \longmapsto & \big\{L(T):T\in\mathbf{T}\textnormal{ with at least one edge}\big\},
        \end{array}$$
    \end{center}
    where $L(T)$ is viewed naturally as a subgraph of $L(G)$.
\end{lemma}

\begin{proof}
    It suffices to observe that there is a correspondence
    \begin{center}
        \vspace{-0.15in}
        $$\begin{array}{crc}
        	\begin{array}{c} \left\{\begin{array}{c}\text{tubes of $G$ with}\\\text{at least one edge}\end{array}\right\} \end{array} & \longleftrightarrow & \begin{array}{c} \left\{\begin{array}{c}\text{induced tubes}\\\text{of $L(G)$}\end{array}\right\} \end{array} \\[3ex]
            T & \longmapsto & L(T),
        \end{array}$$
    \end{center}
    under which overlapping tubes $S,T$ correspond to incompatible tubes $L(S),L(T)$.
\end{proof}

This lemma allows us to translate properties of admissible tubings into properties of complete tubings, so we devote our attention to admissible tubings for the moment. First, we prove a bijection which gives a convenient way to label the tubes of an admissible tubing.

\begin{lemma}\label{lemma_bijection}
    Fix an admissible tubing $\mathbf{T}$ of $G=(V,E)$. For any $T\in\mathbf T$, there is a unique vertex $v\in T$ which is outside every $S\in\mathbf{T}$ with $S\subsetneq T$. Moreover, the map
    $$\begin{array}{crc}
        \mathbf T & \longrightarrow & V \\[1ex]
        T & \longmapsto & v
    \end{array}$$
    is a bijection, so $\mathbf{T}$ has size $n=|V|$.
\end{lemma}

\begin{proof}
    Fix $T\in\mathbf{T}$, and let $\{S_j\}$ be the elements of $\{S\in\mathbf{T}:S\subsetneq T\}$ which are maximal with respect to inclusion. The $\{S_j\}$ are compatible by definition, so there cannot be any edges between them. However, $T$ is connected, so there must be at least one vertex
    $$v\in T\setminus\bigcup_j S_j.$$
    In order to see that $v$ is unique, suppose there are two such vertices. Then the tube induced by
    $$\{v\}\cup\bigcup_{\substack{S_j\text{ incompatible}\\\text{with }\{v\}}}S_j$$
    is strictly smaller than $T$ and contains $v$. This tube cannot belong to $\mathbf T$ (by definition $v$ is outside every $S\subsetneq T$), so we can produce a larger collection of compatible tubes by adjoining it to $\mathbf T$, contradicting maximality.

    In order to see that the map $T\mapsto v$ is a bijection, observe it has an inverse: the map which sends $v\in V$ to the minimal element of $\{T\in\mathbf{T}:v\in T\}$. This set is nonempty by the maximality of $\mathbf{T}$ (for example, $\mathbf{T}$ must contain the connected component of $v$), and it must have a unique minimal element by compatibility (any two tubes containing $v$ must be nested).
\end{proof}

Applying Lemma \ref{lemma_linecorrespondence}, we obtain an immediate corollary.

\begin{corollary}\label{cor_bijection}
    Fix a complete tubing $\mathbf{T}$ of ${G=(V,E)}$. For any $T\in\mathbf T$, there is a unique edge $e\in T$ which is outside every $S\in\mathbf{T}$ with $S\subsetneq T$. Moreover, the map
    $$\begin{array}{crc}
        \Big\{T\in\mathbf{T}\text{ with at least one edge}\Big\} & \longrightarrow & E \\[1ex]
        T & \longmapsto & e
    \end{array}$$
    is a bijection, so $\mathbf{T}$ has size $|V|+|E|$.\qed
\end{corollary}

\begin{remark}
    From now on, whenever we have an admissible tubing $\mathbf T$, we will simplify notation by labeling the vertices $V=\{v_1,\dots,v_n\}$ and tubes $\mathbf T=\{T_1,\dots,T_n\}$ according to the bijection in Lemma \ref{lemma_bijection}, \ie such that $T_i\leftrightarrow v_i$.
\end{remark}

The proof of Lemma \ref{lemma_bijection} suggests that we can think of admissible tubings as being built up by adding vertices one-by-one in some order, merging each with incompatible tubes as we go. The following lemma makes this precise.

\begin{lemma}\label{lemma_order}
    Fix a graph $G=(V,E)$.

    \begin{enumerate}[label=(\alph*)]

        \item
        A total order $<$ on $V$ induces a unique admissible tubing $\{T_1,\dots,T_n\}$ of $G$ which respects $<$, in the sense that
        $$\Big[T_i\subsetneq T_j\Big]\implies\Big[v_i<v_j\Big],$$
        where the vertices are labeled according to the bijection in Lemma \ref{lemma_bijection}. Moreover, every admissible tubing of $G$ arises in this way.

        \item
        Assume $G$ contains no edges from a vertex to itself. An admissible tubing $\{T_1,\dots,T_n\}$ of $G$ induces a unique acyclic orientation of $G$ which respects $\mathbf{T}$, in the sense that
        $$\Big[(v_i\to v_j)\Big]\implies\Big[T_i\subsetneq T_j\Big],$$
        where the vertices are labeled according to the bijection in Lemma \ref{lemma_bijection}. Moreover, every acyclic orientation of $G$ arises in this way.
        
    \end{enumerate}
    
\end{lemma}

\begin{proof}\,
    \begin{enumerate}[label=(\alph*)]

        \item
        Fix a total order
        $$v_1<v_2<\dots<v_n$$
        on $V$. When $n=1$, the claim is trivial.
        
        Inducing on $n$, let $G\setminus v_n$ be the graph obtained by deleting $v_n$ from $G$, and assume we have an admissible tubing $\{T_1,\dots,T_{n-1}\}$ of $G\setminus v_n$ which respects $<$. Letting $T_n$ be the connected component of $v_n$, we obtain an admissible tubing $\{T_1,\dots,T_n\}$ of $G$ which still respects $<$.

        To see that this tubing is unique, suppose that $\mathbf{T}=\{T_1,\dots,T_n\}$ is an admissible tubing of $G$ which respects $<$. Since $v_n>v_j$ for all $j$, $v_n$ cannot be contained in any tube strictly larger than $T_n$. On the other hand, Lemma \ref{lemma_bijection} implies that $v_n$ cannot be contained in any tube strictly smaller than $T_n$. Thus $T_n$ is the only tube containing $v_n$, so it must be the connected component of $v_n$ (otherwise we could add this component to $\mathbf{T}$, violating maximality). Now the rest of the tubes $\{T_1,\dots,T_{n-1}\}$ form an admissible tubing of $G\setminus v_n$ which respects $<$, hence they are uniquely determined by induction.

        To see that every admissible tubing arises from a total order, fix an admissible tubing $\mathbf{T}=\{T_1,\dots,T_n\}$. We can define a partial order on $V$ by
        $$\Big[v_i\leq v_j\Big]\iff\Big[T_i\subseteq T_j\Big],$$
        labeling vertices as usual. Extending this to a total order in any way, we obtain a total order which will induce $\mathbf{T}$.

        \item
        Given an admissible tubing $\{T_1,\dots,T_n\}$ of $G=(V,E)$, we can define a unique acyclic orientation by orienting each $e=\{v_i,v_j\}\in E$ as
        $$\begin{cases}(v_i\to v_j)&\text{if $T_i\subsetneq T_j$},\\(v_i\leftarrow v_j)&\text{if $T_i\supsetneq T_j$},\end{cases}$$
        since compatibility implies that one of these containments holds (by hypothesis, $i\neq j$).
        
        To see that every acyclic orientation arises from an admissible tubing, fix an acyclic orientation of $G$. Consider the partial order on $V$ defined by
        $$\Big[v_i\leq v_j\Big]\iff\Big[\text{there is a path }v_i\to\dots\to v_j\Big].$$
        After extending this to a total order $<$ in any way, part (1) implies there is an admissible tubing which respects $<$, and therefore induces the original orientation.\qedhere

    \end{enumerate}
\end{proof}

\begin{figure}
    \renewcommand{\arraystretch}{1.5}
    \begin{tabular}{ccc}
        $<$ & $\mathbf T$ & $H$
        \\ \hline
        $1<2<3$ &
        \begin{tikzpicture}[baseline=-2]
            \draw[fill=black] (-0.75,0) -- (0.75,0);
            \draw[fill=black] (0,0) circle (.5mm);
            \draw[fill=black] (-0.75,0) circle (.5mm);
            \draw[fill=black] (0.75,0) circle (.5mm);
            \draw[color=gray] (-0.75,0) circle (.15cm and .15cm);
            \node [draw, color=gray, rounded rectangle, minimum height = 1.2em, minimum width = 4em, rounded rectangle arc length = 180] at (-0.375,0) {};
            \node [draw, color=gray, rounded rectangle, minimum height = 1.5em, minimum width = 6.5em, rounded rectangle arc length = 180] at (0,0) {};
        \end{tikzpicture}
        &
        \begin{tikzpicture}[baseline=-2,decoration={markings,mark=at position 0.3 with {\arrow{>}},mark=at position 0.8 with {\arrow{>}}}]
            \draw[postaction={decorate}] (-0.75,0) -- (0.75,0);
            \draw[fill=black] (0,0) circle (.5mm);
            \draw[fill=black] (-0.75,0) circle (.5mm);
            \draw[fill=black] (0.75,0) circle (.5mm);
        \end{tikzpicture}
        \\
        $3<2<1$ &
        \begin{tikzpicture}[baseline=-2]
            \draw[fill=black] (-0.75,0) -- (0.75,0);
            \draw[fill=black] (0,0) circle (.5mm);
            \draw[fill=black] (-0.75,0) circle (.5mm);
            \draw[fill=black] (0.75,0) circle (.5mm);
            \draw[color=gray] (0.75,0) circle (.15cm and .15cm);
            \node [draw, color=gray, rounded rectangle, minimum height = 1.2em, minimum width = 4em, rounded rectangle arc length = 180] at (0.375,0) {};
            \node [draw, color=gray, rounded rectangle, minimum height = 1.5em, minimum width = 6.5em, rounded rectangle arc length = 180] at (0,0) {};
        \end{tikzpicture}
        &
        \begin{tikzpicture}[baseline=-2,decoration={markings,mark=at position 0.3 with {\arrow{<}},mark=at position 0.8 with {\arrow{<}}}]
            \draw[postaction={decorate}] (-0.75,0) -- (0.75,0);
            \draw[fill=black] (0,0) circle (.5mm);
            \draw[fill=black] (-0.75,0) circle (.5mm);
            \draw[fill=black] (0.75,0) circle (.5mm);
        \end{tikzpicture}
        \\
        $1<3<2$
        &
        \begin{tikzpicture}[baseline=-2]
            \draw[fill=black] (-0.75,0) -- (0.75,0);
            \draw[fill=black] (0,0) circle (.5mm);
            \draw[fill=black] (-0.75,0) circle (.5mm);
            \draw[fill=black] (0.75,0) circle (.5mm);
            \draw[color=gray] (-0.75,0) circle (.15cm and .15cm);
            \draw[color=gray] (0.75,0) circle (.15cm and .15cm);
            \node [draw, color=gray, rounded rectangle, minimum height = 1.5em, minimum width = 6.5em, rounded rectangle arc length = 180] at (0,0) {};
        \end{tikzpicture}
        &
        \begin{tikzpicture}[baseline=-2,decoration={markings,mark=at position 0.3 with {\arrow{>}},mark=at position 0.8 with {\arrow{<}}}]
            \draw[postaction={decorate}] (-0.75,0) -- (0.75,0);
            \draw[fill=black] (0,0) circle (.5mm);
            \draw[fill=black] (-0.75,0) circle (.5mm);
            \draw[fill=black] (0.75,0) circle (.5mm);
        \end{tikzpicture}
        \\
        $3<1<2$
        &
        \begin{tikzpicture}[baseline=-2]
            \draw[fill=black] (-0.75,0) -- (0.75,0);
            \draw[fill=black] (0,0) circle (.5mm);
            \draw[fill=black] (-0.75,0) circle (.5mm);
            \draw[fill=black] (0.75,0) circle (.5mm);
            \draw[color=gray] (-0.75,0) circle (.15cm and .15cm);
            \draw[color=gray] (0.75,0) circle (.15cm and .15cm);
            \node [draw, color=gray, rounded rectangle, minimum height = 1.5em, minimum width = 6.5em, rounded rectangle arc length = 180] at (0,0) {};
        \end{tikzpicture}
        &
        \begin{tikzpicture}[baseline=-2,decoration={markings,mark=at position 0.3 with {\arrow{>}},mark=at position 0.8 with {\arrow{<}}}]
            \draw[postaction={decorate}] (-0.75,0) -- (0.75,0);
            \draw[fill=black] (0,0) circle (.5mm);
            \draw[fill=black] (-0.75,0) circle (.5mm);
            \draw[fill=black] (0.75,0) circle (.5mm);
        \end{tikzpicture}
        \\
        $2<1<3$
        &
        \begin{tikzpicture}[baseline=-2]
            \draw[fill=black] (-0.75,0) -- (0.75,0);
            \draw[fill=black] (0,0) circle (.5mm);
            \draw[fill=black] (-0.75,0) circle (.5mm);
            \draw[fill=black] (0.75,0) circle (.5mm);
            \draw[color=gray] (0,0) circle (.15cm and .15cm);
            \node [draw, color=gray, rounded rectangle, minimum height = 1.2em, minimum width = 4em, rounded rectangle arc length = 180] at (-0.375,0) {};
            \node [draw, color=gray, rounded rectangle, minimum height = 1.5em, minimum width = 6.5em, rounded rectangle arc length = 180] at (0,0) {};
        \end{tikzpicture}
        &
        \begin{tikzpicture}[baseline=-2,decoration={markings,mark=at position 0.3 with {\arrow{<}},mark=at position 0.8 with {\arrow{>}}}]
            \draw[postaction={decorate}] (-0.75,0) -- (0.75,0);
            \draw[fill=black] (0,0) circle (.5mm);
            \draw[fill=black] (-0.75,0) circle (.5mm);
            \draw[fill=black] (0.75,0) circle (.5mm);
        \end{tikzpicture}
        \\
        $2<3<1$
        &
        \begin{tikzpicture}[baseline=-2]
            \draw[fill=black] (-0.75,0) -- (0.75,0);
            \draw[fill=black] (0,0) circle (.5mm);
            \draw[fill=black] (-0.75,0) circle (.5mm);
            \draw[fill=black] (0.75,0) circle (.5mm);
            \draw[color=gray] (0,0) circle (.15cm and .15cm);
            \node [draw, color=gray, rounded rectangle, minimum height = 1.2em, minimum width = 4em, rounded rectangle arc length = 180] at (0.375,0) {};
            \node [draw, color=gray, rounded rectangle, minimum height = 1.5em, minimum width = 6.5em, rounded rectangle arc length = 180] at (0,0) {};
        \end{tikzpicture}
        &
        \begin{tikzpicture}[baseline=-2,decoration={markings,mark=at position 0.3 with {\arrow{<}},mark=at position 0.8 with {\arrow{>}}}]
            \draw[postaction={decorate}] (-0.75,0) -- (0.75,0);
            \draw[fill=black] (0,0) circle (.5mm);
            \draw[fill=black] (-0.75,0) circle (.5mm);
            \draw[fill=black] (0.75,0) circle (.5mm);
        \end{tikzpicture}
    \end{tabular}
\caption{Illustration of Lemma \ref{lemma_order} for the $3$-site graph, with vertices $\{1,2,3\}$ from left to right. Six total orders on the vertices induce five distinct admissible tubings, which in turn induce four distinct acyclic orientations.}
\label{fig_orderexample}
\end{figure}
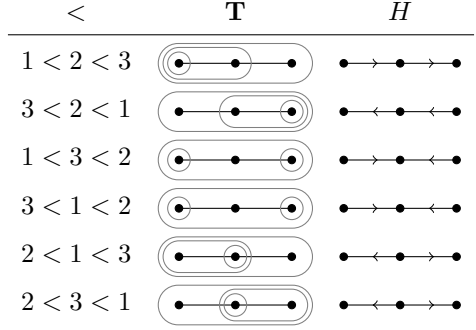

Figure \ref{fig_orderexample} illustrates the lemma for admissible tubings of the $3$-site chain.

\begin{remark}
    Applying Lemma \ref{lemma_linecorrespondence}, it follows that every complete tubing of $G$ is induced by a total order on its edges. In particular, we can list the complete (resp. admissible) tubings of $G$ by listing the total orders on $E$ (resp. $V$), though many total orders may give rise to the same tubing.
\end{remark}

The name ``admissible" is borrowed from \cite{F-P-S-W}, in which a notion of admissible collections of subgraphs is defined that is similar to ours, though not equivalent. For comparison with their definition, we give another characterization of admissible tubings. First, let us extend our definition to directed acyclic graphs.

\begin{definition}
    Let $H$ be a directed acyclic graph with underlying graph $G$. By an \emph{admissible tubing} of $H$, we mean an admissible tubing of $G$ which induces the orientation $H$, in the sense of the previous lemma.
\end{definition}

In this context, we can rephrase the condition of compatibility. For a collection of $n$ induced tubes (recall Definition \ref{def_tubes}) of $H$, being compatible is equivalent to being non-overlapping and each tube having no ``incoming" edges.

\begin{proposition}
    Let $H=(V,E)$ be a directed acyclic graph. A collection $\mathbf{T}$ of induced tubes of $H$ is an admissible tubing of $H$ if and only if:
    \begin{enumerate}

        \item
        $\mathbf{T}$ has size $n=|V|$.

        \item
        $\mathbf{T}$ is non-overlapping.

        \item
        For each tube $T\in\mathbf{T}$, there are no edges $(v\to w)$ with $v\notin T$, $w\in T$.
        
    \end{enumerate}
\end{proposition}

\begin{proof}
    First, suppose that $\mathbf{T}=\{T_1,\dots,T_n\}$ is an admissible tubing of $H$, and label the vertices according to the bijection in Lemma \ref{lemma_bijection}. Condition (1) is implied by Lemma \ref{lemma_bijection}, and (2) by compatibility. To see that (3) is satisfied, fix $T\in\mathbf{T}$ and suppose for a contradiction that there is an incoming edge $(v_i\to v_j)$ for some $v_i\notin T$, $v_j\in T$. It follows that $T_i\not\subseteq T$ and $T_j\subseteq T$. Moreover, since $H$ respects $\mathbf{T}$, we have
    $$(v_i\to v_j)\implies T_i\subseteq T_j.$$
    But then both $T_i\not\subseteq T$ and $T_i\subseteq T_j\subseteq T$, a contradiction.

    Conversely, suppose that $\mathbf{T}=\{T_1,\dots,T_n\}$ satisfies (1)-(3). Maximality is guaranteed by (1) and Lemma \ref{lemma_bijection}. To see that these tubes are compatible, suppose that there is an edge between some non-nested $T_i$, $T_j$. Then (2) implies that $T_i$ and $T_j$ are disjoint, so this edge is incoming for either $T_i$ or $T_j$, contradicting (3).
\end{proof}

\section{Bulk Representation}\label{sec_bulk}

In this section, we show that the flat space wavefunction has an expansion in terms of admissible tubings. Let us introduce our main character.

\begin{definition}\label{def_wavefunction}
    Fix a graph $G=(V,E)$, where $V=\{v_1,\dots,v_n\}$. To $G$, we associate the \emph{flat space wavefunction}\footnote{
    This integral differs slightly from the one appearing in \cite{AH-B-P}; we have dispensed of irrelevant factors of $i=\sqrt{-1}$. The two integrals agree since
    $$\int_{-\infty(1-i\epsilon)}^0 ie^{i\alpha z}\d z=\int_{-\infty}^0 e^{\alpha z}\d z$$
    for $\epsilon\ll0$, $\alpha>0$, as can be seen by making the substitution $iz\mapsto z$ and then integrating along a quarter-circular contour.
    }
    $$\psi_G(\mathbf{X},\mathbf{Y})=\bigintsss_{-\infty}^0\dots\bigintsss_{-\infty}^0\prod_{i=1}^n e^{X_i\eta_i}\prod_{e=\{v_i,v_j\}\in E}P_e(Y_e,\eta_i,\eta_j)\d\eta_1\cdots\d\eta_n,$$
    where $\mathbf{X}=\{X_i\}_{i=1}^n$, $\mathbf{Y}=\{Y_e\}_{e\in E}$ are positive real numbers and
    $$P_e(Y_e,\eta_i,\eta_j)=\frac{1}{2Y_e}\left[e^{-Y_e(\eta_i-\eta_j)}\Theta(\eta_i-\eta_j)+e^{-Y_e(\eta_j-\eta_i)}\Theta(\eta_j-\eta_i)-e^{Y_e(\eta_i+\eta_j)}\right]$$
    is the \emph{propagator} associated to the edge $e=\{v_i,v_j\}\in E$.
\end{definition}

There are some preliminary reductions to be made. First, we split $\psi_G$ into terms corresponding to the spanning subgraphs of $G$.

\begin{definition}
    Fix a graph $G=(V,E)$ with spanning subgraph $H=(V,E_H)$, where $V=\{v_1,\dots,v_n\}$. To the pair $(G,H)$, we associate the function
    $$\psi_{G,H}(\mathbf{X},\mathbf{Y})=\bigintsss_{-\infty}^0\dots\bigintsss_{-\infty}^0\prod_{i=1}^n e^{X_i\eta_i}\prod_{e\in E_H}P_e^+\prod_{e\in E\setminus E_H}P_e^-\d\eta_1\cdots\d\eta_n,$$
    where we have split the propagator into two components:
    \begin{align*}
        P_e^+&=e^{-Y_e(\eta_i-\eta_j)}\Theta(\eta_i-\eta_j)+e^{-Y_e(\eta_j-\eta_i)}\Theta(\eta_j-\eta_i),\\
        P_e^-&=e^{Y_e(\eta_i+\eta_j)},
    \end{align*}
    for each edge $e=\{v_i,v_j\}\in E$.
\end{definition}

\begin{lemma}\label{lemma_sumsubgraphs}
    For any graph $G=(V,E)$,
    $$\psi_G=\frac{1}{\prod\limits_{e\in E}(2Y_e)}\sum_{H\subseteq G}(-1)^{|E\setminus E_H|}\psi_{G,H},$$
    where the sum is taken over all spanning subgraphs $H=(V,E_H)$ of $G$.
\end{lemma}

\begin{proof}
    We can rewrite the product $\prod_{e\in E}P(Y_e,\eta_i,\eta_j)$ in Definition \ref{def_wavefunction} by separating $P_e(Y_e,\eta_i,\eta_j)=\frac{1}{2Y_e}(P_e^+-P_e^-)$ and expanding
    $$\prod_{e\in E}(P_e^+-P_e^-)=\sum_{H\subseteq G}\prod_{e\in E_H}P_e^+\prod_{e\in E\setminus E_H}(-P_e^-).$$
    Collecting the factors of $2Y_e$ and $-1$ gives the desired result.
\end{proof}

Next, we split $\psi_{G,H}$ into terms corresponding to the admissible tubings of $H$.

\begin{definition}
    Fix a graph $G=(V,E)$, a spanning subgraph $H$, and an admissible tubing $\mathbf{T}=\{T_1,\dots,T_n\}$ of $H$. Label the vertices $V=\{v_1,\dots,v_n\}$ according to the bijection in Lemma \ref{lemma_bijection}. To the triple $(G,H,\mathbf{T})$, we associate the function
    $$\psi_{G,H,\mathbf{T}}(\mathbf{X},\mathbf{Y})=\bigintss_{U_{\mathbf{T}}}\prod_{i=1}^n e^{X_i\eta_i}\prod_{e\in E_H}P_e^+\prod_{e\in E\setminus E_H}P_e^-\d\eta_1\cdots\d\eta_n,$$
    where now we integrate over the subset
    $$U_{\mathbf{T}}=\big\{(\eta_1,\dots,\eta_n):\eta_i<\eta_j\text{ for all }T_i\subseteq T_j\big\}\subseteq(-\infty,0)^n,$$
    instead of the entire orthant.
\end{definition}

\begin{lemma}\label{lemma_sumtubings}
    For any graph $G$ and spanning subgraph $H$,
    $$\psi_{G,H}=\sum_{\mathbf{T}\in\Adm(H)}\psi_{G,H,\mathbf{T}}.$$
\end{lemma}

\begin{proof}
    It suffices to show that the collection
    $$\mathcal U=\big\{U_\mathbf{T}:\mathbf{T}\in\Adm(H)\big\}$$
    partitions $(-\infty,0)^n$, up to a measure-zero set. Indeed, fix $(\eta_1,\dots,\eta_n)\in(-\infty,0)^n$. Disregarding the measure-zero set where these coordinates fail to be distinct, we can order them:
    $$\eta_{k_1}<\eta_{k_2}<\dots<\eta_{k_n}.$$
    This induces a total order on $V$, so by Lemma \ref{lemma_order} there is a unique admissible tubing $\mathbf{T}=\{T_{k_1},\dots,T_{k_n}\}$ of $H$ with the property that
    $$\Big[T_{k_i}\subseteq T_{k_j}\Big]\implies\Big[\eta_{k_i}<\eta_{k_j}\Big].$$
    But this is exactly the statement that $(\eta_1,\dots,\eta_n)\in U_\mathbf{T}$. Uniqueness implies $U_\mathbf{T}$ is the only element of $\mathcal U$ containing $(\eta_1,\dots,\eta_n)$.
\end{proof}

\begin{remark}
The preceding lemma is a symptom of the underlying physics, and motivates the name ``bulk" representation (see \cite[Section 2.2]{AH-B-P}). Viewing $\eta_1,\dots,\eta_n$ as times at which events take place (where $\eta=-\infty$ represents the beginning of time and $\eta=0$ represents the present), a natural approach is to decompose $\psi_G$ into terms which correspond to different sequences of events. The problem (and the reason we require admissible tubings) is that such a decomposition is too fine: depending on the structure of $G$, different orders may describe sequences of events which are physically indistinguishable (\ie different total orders on $V$ may induce the same admissible tubing).
\end{remark}

Physics aside, the advantage of restricting the subset we integrate over is that the integrand can now be rewritten in terms of linear forms associated to the tubes in $\mathbf T$ (cf. \cite[Equation 3.17]{AH-B-P}).

\begin{definition}\label{def_linearforms}
    Fix a graph $G=(V,E)$, where $V=\{v_1,\dots,v_n\}$. To any tube $T\in\Tub(G)$, we associate the linear form
    $$\ell_{G,T}(\mathbf{X},\mathbf{Y})=\sum_{v_i\in T}X_i+\sum_{\substack{e=\{v_i,v_j\}\in E,\\v_i\in T,\,v_j\notin T}}Y_e+\sum_{\substack{e=\{v_i,v_j\}\notin T,\\v_i,v_j\in T}}2Y_e.$$
    When there is no possibility of confusion, we will omit $G$ from the subscript and write simply $\ell_T$.
\end{definition}

\begin{remark}\label{rem_overlap}
    These forms may appear somewhat mysterious; see Example \ref{example_bulk} for a concrete case. It is helpful to note that after the change of coordinates
    $$X_i':=X_i+\sum_{e=\{v_i,v_j\}\in E}Y_e,\qquad Y_e':=-2Y_e,$$
    we can think of the linear forms as simply
    $$\ell_T=\sum_{v_i\in T}X_i'+\sum_{e\in T}Y_e'.$$
    In particular, this makes it clear that there are relationships between the $\ell_T$ arising from overlapping tubes (Figure \ref{fig_overlapexample}); these will play a large role in Section \ref{sec_canonicalform}.
\end{remark}

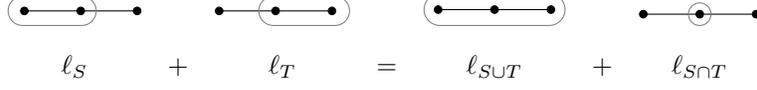
\begin{figure}
    \renewcommand{\arraystretch}{1.5}
    \begin{tabular}{ccccccc}
        \begin{tikzpicture}
            \draw[fill=black] (-0.75,0) -- (0.75,0);
            \draw[fill=black] (0,0) circle (.5mm);
            \draw[fill=black] (-0.75,0) circle (.5mm);
            \draw[fill=black] (0.75,0) circle (.5mm);
            \node [draw, color = gray, rounded rectangle, minimum height = 1.1em, minimum width = 4.0em, rounded rectangle arc length = 180] at (-0.38,0) {};
        \end{tikzpicture}
        & &
        \begin{tikzpicture}
            \draw[fill=black] (-0.75,0) -- (0.75,0);
            \draw[fill=black] (0,0) circle (.5mm);
            \draw[fill=black] (-0.75,0) circle (.5mm);
            \draw[fill=black] (0.75,0) circle (.5mm);
            \node [draw, color = gray, rounded rectangle, minimum height = 1.1em, minimum width = 4.1em, rounded rectangle arc length = 180] at (0.38,0) {};
        \end{tikzpicture}
        & &
        \begin{tikzpicture}
            \draw[fill=black] (-0.75,0) -- (0.75,0);
            \draw[fill=black] (0,0) circle (.5mm);
            \draw[fill=black] (-0.75,0) circle (.5mm);
            \draw[fill=black] (0.75,0) circle (.5mm);
            \node [draw, color = gray, rounded rectangle, minimum height = 1.2em, minimum width = 6em, rounded rectangle arc length = 180] at (0,0) {};
        \end{tikzpicture}
        & &
        \begin{tikzpicture}
            \draw[fill=black] (-0.75,0) -- (0.75,0);
            \draw[fill=black] (0,0) circle (.5mm);
            \draw[fill=black] (-0.75,0) circle (.5mm);
            \draw[fill=black] (0.75,0) circle (.5mm);
            \draw[color=gray] (0,0) circle (.15cm and .15cm);
        \end{tikzpicture}
        \\
        $\ell_S$ & $+$ & $\ell_T$ & $=$ & $\ell_{S\cup T}$ & $+$ & $\ell_{S\cap T}$
    \end{tabular}
\caption{An example of a pair of overlapping tubes $\{S,T\}$ in the $3$-site chain giving rise to a relation among the linear forms.}
\label{fig_overlapexample}
\end{figure}

As promised, let us rewrite $\psi_{G,H,\mathbf T}$ in terms of $\{\ell_T\}_{T\in\mathbf T}$.

\begin{lemma}\label{lemma_tubingcomponent}
    For any graph $G=(V,E)$, spanning subgraph $H=(V,E_H)$, and admissible tubing $\mathbf{T}=\{T_1,\dots,T_n\}$ of $H$,
    $$\psi_{G,H,\mathbf T}=\frac{1}{\prod\limits_{T\in\mathbf{T}}\ell_T}.$$
\end{lemma}

\begin{proof}
    When $(\eta_1,\dots,\eta_n)\in U_\mathbf{T}$, we can simplify
    $$P_e^+=\begin{cases}e^{Y_e(\eta_i-\eta_j)}&\text{if $T_i\subsetneq T_j$},\\e^{-Y_e(\eta_i-\eta_j)}&\text{if $T_j\subsetneq T_i$},\\1&\text{if $T_i=T_j$}\end{cases}$$
    for all $e=\{v_i,v_j\}\in E$, so the integrand of $\psi_{G,H,\mathbf T}$ becomes
    $$\prod_{i=1}^ne^{X_i\eta_i}\prod_{\substack{e\in E_H\\T_i\subsetneq T_j}}e^{Y_e(\eta_i-\eta_j)}\prod_{\substack{e\in E_H\\T_j\subsetneq T_i}}e^{-Y_e(\eta_i-\eta_j)}\prod_{e\in E\setminus E_H}e^{Y_e(\eta_i+\eta_j)}$$
    (in all subscripts, $e=\{v_i,v_j\}$). Now we make the change of coordinates
    $$\eta_i=\sum_{T_i\subseteq T_k}\eta_k',$$
    under which the integrand becomes\footnote{
    This follows from the identity
    $$\sum_{i=1}^n\Bigg[X_i'+\sum_{\substack{e\in E_H\\ T_j\subsetneq T_i}}Y_e'\Bigg]\eta_i=\sum_{k=1}^n\Bigg[\sum_{v_i\in T_k}X_i'+\sum_{e\in T_k}Y_e'\Bigg]\eta_k'.$$
    }
    $$\prod_{i=1}^ne^{X_i\eta_i}\prod_{\substack{e\in E_H\\T_i\subsetneq T_j}}e^{Y_e(\eta_i-\eta_j)}\prod_{\substack{e\in E_H\\T_j\subsetneq T_i}}e^{-Y_e(\eta_i-\eta_j)}\prod_{e\in E\setminus E_H}e^{Y_e(\eta_i+\eta_j)}=\prod_{k=1}^n e^{\ell_{T_k}\eta_k'}.$$
    To see how $U_\mathbf T$ transforms under this coordinate change, note that for $T_i\subseteq T_j$,
    $$\Big[\eta_i<\eta_j\Big]\iff\left[\sum_{T_i\subseteq T_k}\eta_k'<\sum_{T_j\subseteq T_k}\eta_k'\right]\iff\left[\sum_{T_i\subseteq T_k\subsetneq T_j}\eta_k'<0\right].$$
    Hence $\eta_i<\eta_j$ for all $T_i\subseteq T_j$ if and only if $\eta_k'<0$ for all $k$. That is to say,
    $$\Big[(\eta_1,\dots,\eta_n)\in U_{\mathbf T}\Big]\iff\Big[(\eta_1',\dots,\eta_n')\in(-\infty,0)^n\Big].$$
    Thus after the change of coordinates our integral becomes simply
    $$\psi_{G,H,\mathbf T}=\bigintss_{(-\infty,0)^n}\prod_{k=1}^n e^{\ell_{T_k}\eta_k'}\d\eta_1'\cdots\d\eta_n'=\prod_{k=1}^n\int_{-\infty}^0 e^{\ell_{T_k}\eta_k'}\d\eta_k'=\frac{1}{\prod\limits_{T\in\mathbf{T}}\ell_T},$$
    as desired.
\end{proof}

Combining Lemmas \ref{lemma_sumsubgraphs}, \ref{lemma_sumtubings}, and \ref{lemma_tubingcomponent}, we have our main result of the section.

\begin{theorem}\label{theorem_bulk}
    For any graph $G=(V,E)$,
    $$\psi_G=\frac{1}{\prod\limits_{e\in E}(2Y_e)}\sum_{H\subseteq G}\sum_{\mathbf{T}\in\Adm(H)}\frac{(-1)^{|E\setminus E_H|}}{\prod\limits_{T\in\mathbf{T}}\ell_T}.\eqno\qed$$
\end{theorem}

\begin{example}\label{example_bulk}
    Let $G$ be the $1$-loop bubble, with vertices $\{1,2\}$ and edges $\{A,B\}$. To each tube (labeled by its vertex or its edges), we associate a linear form according to Definition \ref{def_linearforms}:
    \begin{center}
    \vspace{0.1in}
        \begin{tabular}{ccc}
            \begin{tikzpicture}
                \draw (0,0) circle (0.375);
                \draw[fill=black] (-0.375,0) circle (.5mm);
                \draw[fill=black] (0.375,0) circle (.5mm);
                \draw[color=gray] (-0.375,0) circle (.15cm and .15cm);
            \end{tikzpicture}
            & &
            \begin{tikzpicture}
                \draw (0,0) circle (0.375);
                \draw[fill=black] (-0.375,0) circle (.5mm);
                \draw[fill=black] (0.375,0) circle (.5mm);
                \draw[color=gray] (0.375,0) circle (.15cm and .15cm);
            \end{tikzpicture}
            \\
            $\ell_1=X_1+Y_A+Y_B$ & \hspace{0.1in} & $\ell_2=X_2+Y_A+Y_B$
        \end{tabular}
    \end{center}
    \begin{center}
        \begin{tabular}{ccccc}
            \begin{tikzpicture}
                \draw (0,0) circle (0.375);
                \draw[fill=black] (-0.375,0) circle (.5mm);
                \draw[fill=black] (0.375,0) circle (.5mm);
                \draw[color=gray] (0:0.175) arc (0:180:0.175);
                \draw[color=gray] (0:0.575) arc (0:180:0.575);
                \draw[color=gray] (180:0.575) arc (-180:0:0.2);
                \draw[color=gray] (0:0.575) arc (0:-180:0.2);
            \end{tikzpicture}
            & &
            \begin{tikzpicture}
                \draw (0,0) circle (0.375);
                \draw[fill=black] (-0.375,0) circle (.5mm);
                \draw[fill=black] (0.375,0) circle (.5mm);
                \draw[color=gray] (0:0.175) arc (0:-180:0.175);
                \draw[color=gray] (0:0.575) arc (0:-180:0.575);
                \draw[color=gray] (180:0.575) arc (180:0:0.2);
                \draw[color=gray] (0:0.575) arc (0:180:0.2);
            \end{tikzpicture}
            & &
            \begin{tikzpicture}
                \draw (0,0) circle (0.375);
                \draw[fill=black] (-0.375,0) circle (.5mm);
                \draw[fill=black] (0.375,0) circle (.5mm);
                \draw[color=gray] (0,0) circle (0.65);
            \end{tikzpicture}
            \\
            $\ell_A=X_1+X_2+2Y_B$ & \hspace{0.1in} & $\ell_B=X_1+X_2+2Y_A$ & \hspace{0.1in} & $\ell_{AB}=X_1+X_2$
        \end{tabular}
    \vspace{0.1in}
    \end{center}
    There are four spanning subgraphs of $G$, and their admissible tubings are:
    \begin{center}
    \vspace{0.1in}
        \begin{tabular}{c|c|c|c}
            $H_1$ & $H_2$ & $H_3$ & $H_4$ \\ \hline & & & \\[-2ex]
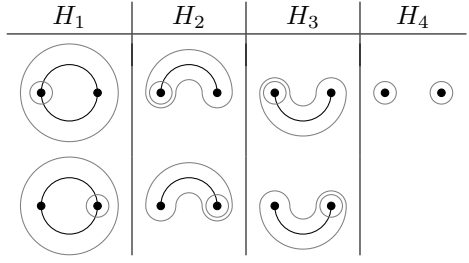
            \begin{tikzpicture}[baseline=1]
                \draw (0,0) circle (0.375);
                \draw[fill=black] (-0.375,0) circle (.5mm);
                \draw[fill=black] (0.375,0) circle (.5mm);
                \draw[color=gray] (-0.375,0) circle (.15cm and .15cm);
                \draw[color=gray] (0,0) circle (0.65);
            \end{tikzpicture}
            &
            \begin{tikzpicture}[baseline=1]
                \draw (0:0.375) arc (0:180:0.375);
                \draw[fill=black] (-0.375,0) circle (.5mm);
                \draw[fill=black] (0.375,0) circle (.5mm);
                \draw[color=gray] (-0.375,0) circle (.15cm and .15cm);
                \draw[color=gray] (0:0.175) arc (0:180:0.175);
                \draw[color=gray] (0:0.575) arc (0:180:0.575);
                \draw[color=gray] (180:0.575) arc (-180:0:0.2);
                \draw[color=gray] (0:0.575) arc (0:-180:0.2);
            \end{tikzpicture}
            &
            \begin{tikzpicture}[baseline=1]
                \draw (0:0.375) arc (0:-180:0.375);
                \draw[fill=black] (-0.375,0) circle (.5mm);
                \draw[fill=black] (0.375,0) circle (.5mm);
                \draw[color=gray] (-0.375,0) circle (.15cm and .15cm);
                \draw[color=gray] (0:0.175) arc (0:-180:0.175);
                \draw[color=gray] (0:0.575) arc (0:-180:0.575);
                \draw[color=gray] (180:0.575) arc (180:0:0.2);
                \draw[color=gray] (0:0.575) arc (0:180:0.2);
            \end{tikzpicture}
            &
            \begin{tikzpicture}[baseline=1]
                \draw[fill=black] (-0.375,0) circle (.5mm);
                \draw[fill=black] (0.375,0) circle (.5mm);
                \draw[color=gray] (-0.375,0) circle (.15cm and .15cm);
                \draw[color=gray] (0.375,0) circle (.15cm and .15cm);
            \end{tikzpicture}
            \\[5ex]
            \begin{tikzpicture}[baseline=1]
                \draw (0,0) circle (0.375);
                \draw[fill=black] (-0.375,0) circle (.5mm);
                \draw[fill=black] (0.375,0) circle (.5mm);
                \draw[color=gray] (0.375,0) circle (.15cm and .15cm);
                \draw[color=gray] (0,0) circle (0.65);
            \end{tikzpicture}
            &
            \begin{tikzpicture}[baseline=1]
                \draw (0:0.375) arc (0:180:0.375);
                \draw[fill=black] (-0.375,0) circle (.5mm);
                \draw[fill=black] (0.375,0) circle (.5mm);
                \draw[color=gray] (0.375,0) circle (.15cm and .15cm);
                \draw[color=gray] (0:0.175) arc (0:180:0.175);
                \draw[color=gray] (0:0.575) arc (0:180:0.575);
                \draw[color=gray] (180:0.575) arc (-180:0:0.2);
                \draw[color=gray] (0:0.575) arc (0:-180:0.2);
            \end{tikzpicture}
            &
            \begin{tikzpicture}[baseline=1]
                \draw (0:0.375) arc (0:-180:0.375);
                \draw[fill=black] (-0.375,0) circle (.5mm);
                \draw[fill=black] (0.375,0) circle (.5mm);
                \draw[color=gray] (0.375,0) circle (.15cm and .15cm);
                \draw[color=gray] (0:0.175) arc (0:-180:0.175);
                \draw[color=gray] (0:0.575) arc (0:-180:0.575);
                \draw[color=gray] (180:0.575) arc (180:0:0.2);
                \draw[color=gray] (0:0.575) arc (0:180:0.2);
            \end{tikzpicture}
            &
        \end{tabular}
    \vspace{0.1in}
    \end{center}
    Hence the bulk representation of $\psi_G$ is
    $$\psi_G=\frac{1}{4Y_A Y_B}\Bigg[\underbrace{\frac{1}{\ell_1\ell_{AB}}+\frac{1}{\ell_2\ell_{AB}}}_{\psi_{G,H_1}}-\underbrace{\left(\frac{1}{\ell_1\ell_A}+\frac{1}{\ell_2\ell_A}\right)}_{\psi_{G,H_2}}-\underbrace{\left(\frac{1}{\ell_1\ell_B}+\frac{1}{\ell_2\ell_B}\right)}_{\psi_{G,H_3}}+\underbrace{\frac{1}{\ell_1\ell_2}}_{\psi_{G,H_4}}\Bigg].$$
\end{example}

\section{Boundary Representation}\label{sec_boundary}

In this section, we show that the flat space wavefunction has an expansion in terms of complete tubings. The key result is proven in \cite[Section 2.3]{AH-B-P}; we fill in the details. The strategy is to show that the flat space wavefunction and the boundary representation both satisfy the same recursion relation, which involves the graphs $G\setminus e$ obtained by deleting an edge $e$ from $G$.

\begin{lemma}\label{lemma_integralrecursion}
    Fix a graph $G=(V,E)$. We have the following recursion relations:
    \begin{enumerate}[label=(\alph*)]

        \item
        If $G$ is a disjoint union of graphs $G_1,G_2$, then
        $$\psi_G=\psi_{G_1}\psi_{G_2}.$$

        \item
        If $G$ is connected and has at least one edge, then
        $$\psi_G=\frac{1}{\ell_G}\sum_{e=\{v_i,v_j\}\in E}\Big[\psi_{G\setminus e}\Big]_{\substack{X_i=X_i+Y_e\\X_j=X_j+Y_e}}$$
        (if $e=\{v_i,v_i\}$ is a loop, we interpret the evaluation as $X_i=X_i+2Y_e$).
        
    \end{enumerate}
\end{lemma}

\begin{proof}\,
    \begin{enumerate}[label=(\alph*)]

        \item
        This is clear: since there are no edges between $G_1$ and $G_2$, the integrand of $\psi_G$ is a product of the integrands of $\psi_{G_1}$ and $\psi_{G_2}$.

        \item
        Fix $k\in\{1,\dots,n\}$. We will use integration by parts one variable at a time:
        \begin{align*}
            \bigintsss_{-\infty}^0\frac{\partial}{\partial\eta_k}&\bigg[{\textstyle\prod_i}e^{X_i\eta_i}\bigg]{\textstyle\prod_e}P_e\d\eta_k\\
            &=\bigg[{\textstyle\prod_i}e^{X_i\eta_i}{\textstyle\prod_e}P_e\bigg]_{\eta_k=-\infty}^{\eta_k=0}-\bigintsss_{-\infty}^0{\textstyle\prod_i}e^{X_i\eta_i}\frac{\partial}{\partial\eta_k}\bigg[{\textstyle\prod_e}P_e\bigg]\d\eta_k.
        \end{align*}
        Since $G$ is connected and has at least one edge, there must be some edge $e=\{v_i,v_k\}\in E$. For this edge,
        $$P_e\big|_{\eta_k=0}=\frac{1}{2Y_e}\left[e^{-Y_e(\eta_i)}\Theta(\eta_i)+e^{-Y_e(-\eta_i)}\Theta(-\eta_i)-e^{Y_e(\eta_i)}\right]=0$$
        (in the region of interest, where $\eta_i<0$), and as a consequence the boundary term vanishes. The integration by parts now reads
        $$\bigintsss_{-\infty}^0\frac{\partial}{\partial\eta_k}\bigg[{\textstyle\prod_i}e^{X_i\eta_i}\bigg]{\textstyle\prod_e}P_e\d\eta_k=-\bigintsss_{-\infty}^0{\textstyle\prod_i}e^{X_i\eta_i}\frac{\partial}{\partial\eta_k}\bigg[{\textstyle\prod_e}P_e\bigg]\d\eta_k,$$
        and we can reintroduce the rest of the integrals and sum over $k$:
        \begin{align*}
            \bigintsss_{-\infty}^0\cdots\bigintsss_{-\infty}^0&\delta\bigg[{\textstyle\prod_i}e^{X_i\eta_i}\bigg]{\textstyle\prod_e}P_e\d\eta_1\cdots\d\eta_n\\
            &=-\bigintsss_{-\infty}^0\cdots\bigintsss_{-\infty}^0{\textstyle\prod_i}e^{X_i\eta_i}\delta\bigg[{\textstyle\prod_e}P_e\bigg]\d\eta_1\cdots\d\eta_n,
        \end{align*}
        where $\delta$ is the differential operator $\frac{\partial}{\partial\eta_1}+\cdots+\frac{\partial}{\partial\eta_n}$. Finally, we compute
        \begin{align*}
            \delta\bigg[{\textstyle\prod_i}e^{X_i\eta_i}\bigg]{\textstyle\prod_e}P_e&=(X_1+\dots+X_n){\textstyle\prod_i}e^{X_i\eta_i}{\textstyle\prod_e}P_e\\
            &=\ell_G\bigg[{\textstyle\prod_i}e^{X_i\eta_i}{\textstyle\prod_e}P_e\bigg]
        \end{align*}
        (since $G$ is connected), and
        \begin{align*}
            {\textstyle\prod_i}e^{X_i\eta_i}\delta\bigg[{\textstyle\prod_e}P_e\bigg]&={\textstyle\prod_i}e^{X_i\eta_i}\sum_{e=\{v_j,v_k\}\in E}\left[\delta(P_e)\prod_{e'\in E\setminus\{e\}}P_{e'}\right]\\
            &=\sum_{e=\{v_j,v_k\}\in E}\left[{\textstyle\prod_i}e^{X_i\eta_i}\bigg[-e^{Y_e(\eta_j+\eta_k)}\bigg]\prod_{e'\in E\setminus\{e\}}P_{e'}\right]\\
            &=-\sum_{e=\{v_j,v_k\}\in E}\left[{\textstyle\prod_i}e^{X_i\eta_i}\prod_{e'\in E\setminus\{e\}}P_{e'}\right]_{\substack{X_j=X_j+Y_e\\X_k=X_k+Y_e}}.
        \end{align*}
        Substituting these into the integrals above, we recognize the left hand side as $\ell_G\psi_G$ and the right hand side as
        $$\sum_{e=\{v_j,v_k\}\in E}\Big[\psi_{G\setminus e}\Big]_{\substack{X_j=X_j+Y_e\\X_k=X_k+Y_e}}.$$
        Dividing both sides by $\ell_G$ gives the desired result.\qedhere

    \end{enumerate}
\end{proof}

\begin{lemma}\label{lemma_completerecursion}
    Fix a graph $G=(V,E)$. We have the following recursion relations:
    \begin{enumerate}[label=(\alph*)]

        \item
        If $G$ is a disjoint union of graphs $G_1,G_2$, then
        $$\sum_{\mathbf T\in\Com(G)}\frac{1}{\prod\limits_{T\in\mathbf T}\ell_{G,T}}=\left[\sum_{\mathbf T_1\in\Com(G_1)}\frac{1}{\prod\limits_{T\in\mathbf T_1}\ell_{G_1,T}}\right]\left[\sum_{\mathbf T_2\in\Com(G_2)}\frac{1}{\prod\limits_{T\in\mathbf T_2}\ell_{G_2,T}}\right]$$

        \item
        If $G$ is connected and has at least one edge, then
        $$\sum_{\mathbf T\in\Com(G)}\frac{1}{\prod\limits_{T\in\mathbf T}\ell_{G,T}}=\frac{1}{\ell_{G,G}}\sum_{e=\{v_i,v_j\}\in E}\left[\sum_{\mathbf T_e\in\Com(G\setminus e)}\frac{1}{\prod\limits_{T\in\mathbf T_e}\ell_{G\setminus e,T}}\right]_{\substack{X_i=X_i+Y_e\\X_j=X_j+Y_e}}$$
        (if $e=\{v_i,v_i\}$ is a loop, we interpret the evaluation as $X_i=X_i+2Y_e$).

    \end{enumerate}
\end{lemma}

\begin{proof}\,
    \begin{enumerate}[label=(\alph*)]

        \item
        First, in order to see that
        $$\sum_{\mathbf T\in\Com(G)}\frac{1}{\prod\limits_{T\in\mathbf T}\ell_{G,T}}=\left[\sum_{\mathbf T_1\in\Com(G_1)}\frac{1}{\prod\limits_{T\in\mathbf T_1}\ell_{G,T}}\right]\left[\sum_{\mathbf T_2\in\Com(G_2)}\frac{1}{\prod\limits_{T\in\mathbf T_2}\ell_{G,T}}\right],$$
        note that $\mathbf{T}$ is a complete tubing of $G$ if and only if it is a disjoint union of a complete tubing $\mathbf{T}_1$ of $G_1$ and a complete tubing $\mathbf{T}_2$ of $G_2$.
        
        Now, the desired equality follows from the fact that
        $$\ell_{G,T}=\ell_{G_1,T}$$
        for all $T\in\Tub(G_1)$, and likewise for $G_2$.

        \item
        First, in order to see that
        $$\sum_{\mathbf T\in\Com(G)}\frac{1}{\prod\limits_{T\in\mathbf T}\ell_{G,T}}=\frac{1}{\ell_{G,G}}\sum_{e\in E}\left[\sum_{\mathbf T_e\in\Com(G\setminus e)}\frac{1}{\prod\limits_{T\in\mathbf T_e}\ell_{G,T}}\right],$$
        note that $\mathbf{T}$ is a complete tubing of $G$ if and only if it is a disjoint union of $\{G\}$ and a complete tubing $\mathbf{T}_e$ of $G\setminus e$, for some unique $e\in E$ (since $G$ is connected, any complete tubing $\mathbf{T}$ of $G$ must contain the tube $G$, hence by Corollary \ref{cor_bijection} there is a unique edge $e$ outside every other tube).

        Now, the desired equality follows from the fact that for $e=\{v_i,v_j\}\in E$,
        $$\ell_{G,T}=\Big[\ell_{G\setminus e,T}\Big]_{\substack{X_i=X_i+Y_e\\X_j=X_j+Y_e}}$$
        for all $T\in\Tub(G\setminus e)$.\qedhere

    \end{enumerate}
\end{proof}

\begin{theorem}\label{theorem_boundary}
    For any graph $G$,
    $$\psi_G=\sum_{\mathbf{T}\in\Com(G)}\frac{1}{\prod\limits_{T\in\mathbf{T}}\ell_T}.$$
\end{theorem}

\begin{proof}
    Lemmas \ref{lemma_integralrecursion} and \ref{lemma_completerecursion} establish that both expressions satisfy the same recursion relations. We need only to check the base case, where $G$ is a connected graph with no edges, \ie the graph with a single vertex $v_1$. In this case, there is only one complete tubing $\mathbf{T}=\{G\}$, and
    $$\psi_G=\int_{-\infty}^0 e^{X_1\eta_1}\d\eta_1=\frac{1}{X_1}=\frac{1}{\ell_G},$$
    as desired.
\end{proof}

\begin{example}\label{example_boundary}
    Let $G$ be the $1$-loop bubble from Example \ref{example_bulk}. There are two complete tubings of $G$:
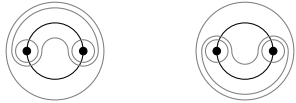
    \begin{center}
        \begin{tabular}{ccc}
            \begin{tikzpicture}
                \draw (0,0) circle (0.375);
                \draw[fill=black] (-0.375,0) circle (.5mm);
                \draw[fill=black] (0.375,0) circle (.5mm);
                \draw[color=gray] (-0.375,0) circle (.15cm and .15cm);
                \draw[color=gray] (0.375,0) circle (.15cm and .15cm);
                \draw[color=gray] (0,0) circle (0.65);
                \draw[color=gray] (0:0.175) arc (0:180:0.175);
                \draw[color=gray] (0:0.575) arc (0:180:0.575);
                \draw[color=gray] (180:0.575) arc (-180:0:0.2);
                \draw[color=gray] (0:0.575) arc (0:-180:0.2);
            \end{tikzpicture}
            & \hspace{0.2in} &
            \begin{tikzpicture}
                \draw (0,0) circle (0.375);
                \draw[fill=black] (-0.375,0) circle (.5mm);
                \draw[fill=black] (0.375,0) circle (.5mm);
                \draw[color=gray] (-0.375,0) circle (.15cm and .15cm);
                \draw[color=gray] (0.375,0) circle (.15cm and .15cm);
                \draw[color=gray] (0,0) circle (0.65);
                \draw[color=gray] (0:0.175) arc (0:-180:0.175);
                \draw[color=gray] (0:0.575) arc (0:-180:0.575);
                \draw[color=gray] (180:0.575) arc (180:0:0.2);
                \draw[color=gray] (0:0.575) arc (0:180:0.2);
            \end{tikzpicture}
        \end{tabular}
    \end{center}
    Hence the boundary representation of $\psi_G$ is
    $$\psi_G=\frac{1}{\ell_1\ell_2\ell_A\ell_{AB}}+\frac{1}{\ell_1\ell_2\ell_B\ell_{AB}}.$$
\end{example}

\section{Canonical Form Representation}\label{sec_canonicalform}

In this section, we show that the flat space wavefunction can be read off from the canonical form of the cosmological polytope. To be precise, we will show that the canonical form representation agrees with the boundary representation. Let us introduce the relevant players.

\begin{proposition}[see \cite{AH-B-P,K-M,AH-B-L}]\label{prop_polytope}
    To any graph $G=(V,E)$, there is an associated:
    \begin{enumerate}[label=-]

        \item
        \emph{cosmological polytope} $P_G$, a polytope in $\R^{|V|+|E|}$ (with coordinate variables $\{X_i\}$, $\{Y_e\}$) with facet hyperplanes $\{\ell_T=0\}_{T\in\Tub(G)}$,
    
        \item
        \emph{dual cosmological polytope} $P_G^\lor$, a polytope with vertices corresponding to the facets of $P_G$,
        
        \item
        \emph{canonical form} $\Omega_G$, a top-dimensional differential form which has logarithmic singularities on the facets of $P_G$ and is regular everywhere else.
        \hfill\Diamond
        
    \end{enumerate}
\end{proposition}

The cosmological polytope was introduced in \cite{AH-B-P}; a precise description of its faces is given in \cite{K-M}. For an overview on canonical forms, which in general are associated to a \emph{positive geometry} (of which polytopes are an example), see \cite{AH-B-L}. We give only a brief review of the relevant results.

Consider a polytope $P$ with facet hyperplanes $\{F_i=0\}_{i=1}^m$. Moreover, suppose each $F_i$ has integer coefficients (as in the case of $P_G$).

\begin{proposition}[{\cite[Theorem 3.2]{Gaetz}}]\label{prop_canonicalform}
    The canonical form of $P$ is
    $$\Omega_P=\frac{\adj_{P^\lor}(\mathbf x)}{\prod_i F_i}\d\mathbf{x}.$$
    where $\adj_{P^\lor}(\mathbf x)$ is the adjoint polynomial of the dual polytope $P^\lor$.
    \qed
\end{proposition}

The adjoint polynomial was first defined for polytopes by Warren \cite{Warren}.

\begin{definition}\label{def_adjointpolynomial}
    Fix a triangulation $\{\sigma_1,\dots,\sigma_r\}$ of the dual polytope $P^\lor$. The \emph{adjoint polynomial} of $P^\lor$ is\footnote{
    It may seem strange to define $\adj_{P^\lor}$ rather than $\adj_P$, but in some sense it is more natural. For example, the vanishing locus of $\adj_{P^\lor}$ is the \emph{adjoint hypersurface} of $P$ (not of $P^\lor$!).
    }
    $$\adj_{P^\lor}(\mathbf x)=\sum_{j=1}^r\vol(\sigma_j)\prod_{F_i\notin\sigma_j}F_i(\mathbf x),$$
    where $\vol(\sigma_j)$ is the normalized volume of $\sigma_j$ (normalized so that the unit simplex has volume 1) and by ``$F_i\notin\sigma_j$" we mean ``the vertex of $P^\lor$ corresponding to $F_i$ is not in $\sigma_j$".
\end{definition}

Though the adjoint polynomial appears to depend on a choice of triangulation, in \cite{Warren} Warren shows that it is independent of this choice. However, note that the adjoint polynomial is sensitive to rescalings of the $F_i$, and therefore is only defined up to a constant. Different normalizations appear in the literature; this is of no concern to us because the ratio
$$\frac{\adj_{P^\lor}}{\prod_i F_i}$$
appearing in Proposition \ref{prop_canonicalform} and Theorem \ref{theorem_canonicalform} is well-defined regardless.

The difficulty of computing the adjoint polynomial via Definition \ref{def_adjointpolynomial} is that we require a triangulation of $P^\lor$. Fortunately, there is a correspondence between the regular triangulations of a polytope and the initial ideals of its toric ideal:

\begin{definition}
    The \emph{toric ideal} of $P^\lor$ is
    $$I_{P^\lor}=\big(\mathbf t^{\mathbf u}-\mathbf t^{\mathbf v}:{\textstyle\sum_i u_i F_i=\sum_i v_i F_i}\big)\subseteq k[t_1,\dots,t_m],$$
    where we have used the multi-index notation $\mathbf{t}^{\mathbf u}:=t_1^{u_1}\cdots t_m^{u_m}$.
\end{definition}

To be precise, $I_{P^\lor}$ is the toric ideal defined by the coefficient vectors $\mathbf f_i\in\mathbb Z^d$ of $F_i$ (\ie $F_i=f_{i1}x_1+\dots+f_{id}x_d$). Moreover, there is a natural multigrading on $k[t_1,\dots,t_m]$ defined by $\deg(t_i)=\mathbf f_i$, which we will make use of.

\begin{proposition}[{\cite[Corollary 8.4 and Theorem 8.8]{Sturmfels}}]\label{prop_triangulationcorrespondence}
    Fix a monomial order $<$ on $k[t_1,\dots,t_m]$. If the minimal associated primes $\{\p_{K_1},\dots,\p_{K_r}\}$ of $\init_<(I_{P^\lor})$ are
    $$\p_{K_j}=(t_k:k\in K_j),\qquad K_j\subseteq\{1,\dots,m\},$$
    then $\{\sigma_1,\dots,\sigma_r\}$ is a regular triangulation of $P^\lor$, where $\sigma_j$ is the simplex with vertices corresponding to $\{F_k:k\notin K_j\}$. Moreover, the normalized volume of $\sigma_j$ is the multiplicity of $I_{P^\lor}$ at $\p_{K_j}$.
    \qed
\end{proposition}

The triangulation we obtain depends on the choice of monomial order. Returning to the setting of the cosmological polytope (recall $P_G$ is the polytope with facet hyperplanes $\{\ell_T=0\}$), the toric ideal of $P_G^\lor$ is
$$I_G=\big(\mathbf t^{\mathbf u}-\mathbf t^{\mathbf v}:{\textstyle\sum_T u_T\ell_T=\sum_T v_T\ell_T}\big)\subseteq k\big[t_T:T\in\Tub(G)\big],$$
and it will be convenient to use the reverse lexicographic\footnote{Here we must use the natural multigrading on $k[t_T:T\in\Tub(G)]$, not the standard grading.} order $<$ induced by any total order on $\{t_T\}$ satisfying
$$t_T<t_{T'}\text{ for all }T\supsetneq T'.$$
We will see that with this order, the minimal associated primes of $\init_<(I_G)$ exactly correspond to complete tubings of $G$.

\begin{lemma}\label{lemma_init}
    For any graph $G$,
    $$\init_<(I_G)=\big(t_S t_T:\textnormal{$S,T\in\Tub(G)$ overlapping}\big),$$
    where $<$ is the order described in the previous paragraph.
\end{lemma}

\begin{proof}
    For each pair of overlapping $S,T\in\Tub(G)$, we can write $S\cap T$ as a disjoint union $\bigcup_i U_i$ of tubes. Then, the equality (recall Remark \ref{rem_overlap})
    $$\ell_S+\ell_T=\ell_{S\cup T}+\textstyle\sum_i\ell_{U_i}$$
    implies that the binomial $t_S t_T-t_{S\cup T}\textstyle\prod_i t_{U_i}$ belongs to $I_G$. We claim these binomials form a Gr\"obner basis\footnote{We note the similarity to the so-called ``Hibi relations" which form a Gr\"obner basis for toric ideals associated to distributive lattices \cite{Hibi}.}
    $$\mathcal G:=\big\{t_S t_T-t_{S\cup T}\textstyle\prod_i t_{U_i}:\text{$S,T\in\Tub(G)$ overlapping}\big\}$$
    for $I_G$ with respect to $<$. This will imply the desired result, as $t_S t_T$ is the leading monomial of $t_S t_T-t_{S\cup T}\textstyle\prod_i t_{U_i}$ (since $t_{S\cup T}$ is smaller than $t_S$ and $t_T$).

    To show that $\mathcal G$ is a Gr\"obner basis of $I_G$, it suffices to show that for any primitive\footnote{
    A binomial $\mathbf t^{\mathbf u}-\mathbf t^{\mathbf v}$ in a toric ideal $I$ is \emph{primitive} if there are no other binomials $\mathbf t^{\mathbf r}-\mathbf t^{\mathbf s}\in I$ with $\mathbf t^{\mathbf r}\mid\mathbf t^{\mathbf u}$ and $\mathbf t^{\mathbf s}\mid\mathbf t^{\mathbf v}$. The ``suffices to show" follows from the fact that the set of all primitive binomials is always a Gr\"obner basis \cite[Proposition 4.11]{Sturmfels}.}
    binomial $\textstyle\prod_i t_{S_i}-\prod_j t_{T_j}$ (with leading monomial $\textstyle\prod_i t_{S_i}$) in $I_G$, the collection $\{S_i\}$ contains a pair of overlapping tubes. Indeed: consider the smallest monomial among $\{t_{T_j}\}$, say $t_{T_0}$, and fix a vertex $v_0\in T_0$. The equality
    $$\textstyle\sum_i\ell_{S_i}=\sum_j\ell_{T_j}$$
    implies that $v_0\in S_i$ for some $i$. Fix a maximal element of $\{S_i:v_0\in S_i\}$, say $S_0$. There are three cases.

    $S_0\supsetneq T_0$: this implies $t_{S_0}<t_{T_0}$, but $t_{T_0}$ was the smallest monomial among $\{t_{T_j}\}$. Hence $\textstyle\prod_i t_{S_i}<\prod_j t_{T_j}$, contradicting $\textstyle\prod_i t_{S_i}$ being the leading monomial.

    $S_0=T_0$: this contradicts $\textstyle\prod_i t_{S_i}-\prod_j t_{T_j}$ being primitive.

    $S_0\not\supseteq T_0$: this implies there is at least one edge $e=\{v,w\}\in T_0$ which leaves $S_0$, \ie with $v\in S_0$, $e\notin S_0$. Again, the equality
    $$\textstyle\sum_i\ell_{S_i}=\sum_j\ell_{T_j}$$
    implies $e\in S_i$ for some $i$, say $S_1$. Now $S_0$ and $S_1$ must be overlapping:
    \begin{enumerate}[label=-]
    
        \item
        the intersection $S_0\cap S_1$ is nonempty since it contains $v$,
        
        \item
        we cannot have $S_1\subseteq S_0$ since $e\notin S_0$,
        
        \item
        we cannot have $S_1\supsetneq S_0$ since $S_0$ was a maximal element of $\{S_i:v_0\in S_i\}$.
        
    \end{enumerate}
    Having found a pair of overlapping tubes in $\{S_i\}$, we conclude that $\mathcal G$ is a Gr\"obner basis of $I_G$.
\end{proof}

\begin{lemma}\label{lemma_triangulation}
    For any graph $G$, $\{\sigma_{\mathbf T}:\mathbf T\in\Com(G)\}$ is a regular unimodular\footnote{A \emph{unimodular} triangulation is one in which every simplex has normalized volume $1$.} triangulation of $P_G^\lor$, where $\sigma_{\mathbf T}$ is the simplex with vertices corresponding to ${\{\ell_T:T\in\mathbf T\}}$.
\end{lemma}

\begin{proof}
    In view of Lemma \ref{lemma_init}, all that is left is to determine the minimal associated primes of
    $$\init_<(I_G)=\big(t_S t_T:\textrm{$S,T\in\Tub(G)$ overlapping}\big).$$
    In fact, this is the edge ideal of the graph $\widetilde G$ which has a vertex for each tube $T\in\Tub(G)$ and an edge for each pair of overlapping tubes $S,T\in\Tub(G)$. For edge ideals, there is a correspondence (\cite[Lemma 9.1.4]{H-H})
    \begin{center}
        \vspace{-0.2in}
        $$\begin{array}{ccc}
                \begin{array}{c} \left\{\begin{array}{c}\text{minimal associated}\\\text{primes of }\init_<(I_G)\end{array}\right\} \end{array} & \longleftrightarrow & \begin{array}{c} \left\{\begin{array}{c}\text{maximal independent}\\\text{sets of }\widetilde G\end{array}\right\} \end{array} \\[3ex]
            (t_T:T\notin\mathbf{T}) & \longleftarrow\!\shortmid & \mathbf{T}.
        \end{array}$$
    \end{center}
    (each prime has multiplicity 1). By maximal independent sets, we mean maximal collections of vertices of $\widetilde G$ with no edges between them; \ie maximal collections of tubes of $G$ with no two tubes overlapping. But this is exactly what we have been calling complete tubings of $G$! Hence the minimal associated primes of $\init_<(I_G)$ are
    $$\p_{\mathbf{T}}=(t_T:T\notin\mathbf{T}),\qquad\mathbf{T}\in\Com(G),$$
    and the result follows from Proposition \ref{prop_triangulationcorrespondence}.
\end{proof}

Triangulation in hand, we can at last obtain the canonical form representation.

\begin{theorem}\label{theorem_canonicalform}
    For any graph $G$,
    $$\frac{\adj_G}{\displaystyle\prod\limits_{T\in\Tub(G)}\ell_T}=\sum_{\mathbf{T}\in\Com(G)}\frac{1}{\prod\limits_{T\in\mathbf{T}}\ell_T}.$$
\end{theorem}

\begin{proof}
    Applying Definition \ref{def_adjointpolynomial} with the triangulation in Lemma \ref{lemma_triangulation}, we have
    $$\adj_G=\sum_{\mathbf T\in\Com(G)}\vol(\sigma_{\mathbf T})\prod_{\ell_T\notin\sigma_{\mathbf T}}\ell_T=\sum_{\mathbf{T}\in\Com(G)}\prod_{T\notin\mathbf{T}}\ell_T.$$
    Dividing both sides by $\prod_{T\in\Tub(G)}\ell_T$ gives the desired result.
\end{proof}

\begin{example}
    Let $G$ be the $1$-loop bubble from Examples \ref{example_bulk} and \ref{example_boundary}. The complete tubings of $G$ are $\{\{1\},\{2\},\{A\},\{AB\}\}$ and $\{\{1\},\{2\},\{B\},\{AB\}\}$, so the triangulation of $P_G^\lor$ is the one shown in Figure \ref{fig_polytopeexample}, its adjoint polynomial is
    $$\adj_G=\ell_A+\ell_B,$$
    and the canonical form representation of $\psi_G$ is
    $$\psi_G=\frac{\ell_A+\ell_B}{\ell_1\ell_2\ell_A\ell_B\ell_{AB}}.$$
\end{example}

\begin{figure}
    \begin{tabular}[t]{ccc}
        \resizebox{0.3\textwidth}{!}{
        \begin{tikzpicture}
            \node (A) at (0,4) {};
            \node (B) at (0,-4) {};
            \node (1) at (-3,0.5) {};
            \node (2) at (3,0.5) {};
            \node (AB) at (0.25,-0.5) {};
            \draw[dashed] (1) -- (2);
            \draw (1) -- (A);
            \draw (1) -- (B);
            \draw (1) -- (AB);
            \draw (2) -- (A);
            \draw (2) -- (B);
            \draw (2) -- (AB);
            \draw (A) -- (AB);
            \draw (B) -- (AB);
            \fill[draw, color=black, fill=blue, opacity=0.2] (1.center) -- (A.center) -- (AB.center);
            \fill[draw, color=black, fill=blue, opacity=0.2] (2.center) -- (A.center) -- (AB.center);
            \fill[draw, color=black, fill=red, opacity=0.2] (1.center) -- (B.center) -- (AB.center);
            \fill[draw, color=black, fill=red, opacity=0.2] (2.center) -- (B.center) -- (AB.center);
            \node[color=blue] at (A) [circle,draw,fill=white,inner sep=2pt] {$A$};
            \node[color=red] at (B) [circle,draw,fill=white,inner sep=2pt] {$B$};
            \node[color=gray] at (1) [circle,draw,fill=white,inner sep=2pt] {$1$};
            \node[color=gray] at (2) [circle,draw,fill=white,inner sep=2pt] {$2$};
            \node[color=gray] at (AB) [circle,draw,fill=white,inner sep=2pt] {$AB$};
        \end{tikzpicture}
        }
        & &
        \begin{tabular}{c}
            \\[-2in]
            \begin{tikzpicture}
                \draw (0,0) circle (0.375);
                \draw[fill=black] (-0.375,0) circle (.5mm);
                \draw[fill=black] (0.375,0) circle (.5mm);
                \draw[color=gray] (-0.375,0) circle (.15cm and .15cm);
                \draw[color=gray] (0.375,0) circle (.15cm and .15cm);
                \draw[color=gray] (0,0) circle (0.65);
                \draw[color=blue] (0:0.175) arc (0:180:0.175);
                \draw[color=blue] (0:0.575) arc (0:180:0.575);
                \draw[color=blue] (180:0.575) arc (-180:0:0.2);
                \draw[color=blue] (0:0.575) arc (0:-180:0.2);
            \end{tikzpicture}
            \\[0.2in]
            \begin{tikzpicture}
                \draw (0,0) circle (0.375);
                \draw[fill=black] (-0.375,0) circle (.5mm);
                \draw[fill=black] (0.375,0) circle (.5mm);
                \draw[color=gray] (-0.375,0) circle (.15cm and .15cm);
                \draw[color=gray] (0.375,0) circle (.15cm and .15cm);
                \draw[color=gray] (0,0) circle (0.65);
                \draw[color=red] (0:0.175) arc (0:-180:0.175);
                \draw[color=red] (0:0.575) arc (0:-180:0.575);
                \draw[color=red] (180:0.575) arc (180:0:0.2);
                \draw[color=red] (0:0.575) arc (0:180:0.2);
            \end{tikzpicture}
        \end{tabular}
    \end{tabular}
\caption{Illustration of the triangulation of the dual cosmological polytope for the $1$-loop bubble. Each vertex of the polytope corresponds to a tube, and each simplex in the triangulation corresponds to a complete tubing.}
\label{fig_polytopeexample}
\end{figure}
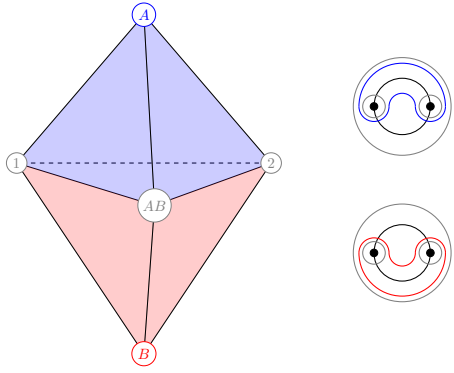

\section*{Acknowledgements}

I would like to thank Anna-Laura Sattelberger, Daniel Bath, and Uli Walther for many helpful conversations, as well as Claudia Fevola for providing useful input.

\subsection*{Data Availability} No data was generated or analyzed in this work.

\subsection*{Declarations} This work was supported in part by NSF grant DMS-2100288 and by Simons Foundation Collaboration Grant for Mathematicians \#580839 and SFI-MPS-TSM-00012928. The author has no competing interests to declare.

\bibliographystyle{plain}
\bibliography{references}

\end{document}